\documentclass[twocolumn]{svjour3}

\smartqed
\usepackage[T1]{fontenc}      
\usepackage{amsmath}
\usepackage{amsfonts}
\usepackage{epsfig}
\usepackage{booktabs} 

\usepackage{pstricks,float,fancybox,amssymb,amsmath,graphicx,t1enc,epsfig,psfrag}\usepackage{subfigure}

\newenvironment{bulletitemize}{%
\begin{itemize}}{\end{itemize}}

\newcommand{\ra}[1]{\renewcommand{\arraystretch}{#1}} 
\newcommand{\MinSLDPC}{\textsc{Min-SLDPC }}
\newcommand{\MinSLDPCF}{\textsc{Min-SLDPC}}
\newcommand{\MinDPC}{\textsc{Min-DPC }}
\newcommand{\MinDPCF}{\textsc{Min-DPC}}
\newcommand{\SLDPC}{\textsc{SLDPC }}
\newcommand{\SLDPCF}{\textsc{SLDPC}}

\newcommand{\graphicsdir}{}

\begin{document}

\title{Isomorphic coupled-task scheduling problem with compatibility constraints on a single processor}
%BD: ai modifie les noms pour tenir sur une ligne (suppression du prenom, rajout de l'initiale
\author{G. Simonin \and B. Darties \and R. Giroudeau \and J.-C. K\"onig} 

\institute{G. Simonin \and R. Giroudeau \and J.-C. K\"onig \at
		LIRMM UMR 5506, rue Ada,\\
		34392 Montpellier Cedex 5 - France\\
        	\email{\{simonin,rgirou,konig\}@lirmm.fr}
           \and
%BD :  modifi mes coordonnes perso
           B. Darties \at  
        	LE2I UMR 5158\\
		9 Rue Alain Savary\\
		21000 Dijon -  France\\
        	\email{Benoit.Darties@u-bourgogne.fr}     		
}

\date{Received: date / Accepted: date}

\maketitle

\begin{abstract}
The problem presented in this paper is a generalization of the usual coupled-tasks scheduling pro\-blem in presence of compatibility constraints. The reason behind this study is the data acquisition problem for a submarine torpedo. We investigate a particular configuration for coupled-tasks (any task is divided into two sub-tasks separated by an idle time), in which the idle time of a coupled-task is equal to the sum of durations of its two sub-tasks. We prove $\mathcal{NP}$-completeness of the minimization of the schedule length, we show that finding a solution to our problem amounts to solving a graph problem, which in itself is close to the minimum-disjoint path cover (min-DCP) problem. We design a  $\left(\frac{3a+2b}{2a+2b}\right)$- approximation, where $a$ and $b$ (the processing time of the two sub-tasks) are two input data such as $a>b>0$, and that leads to a ratio between $\frac{3}{2}$ and $\frac{5}{4}$. Using a polynomial-time algorithm developed for some class of graph of min-DCP, we show that the ratio decreases to $\frac{1+\sqrt{3}}{2}\approx 1.37$.
\end{abstract}

\section{Introduction}
In this paper, we present a scheduling problem of coupled-tasks subject to compatibility constraints, which is a generalization of the scheduling problem of coupled-tasks first introduced by Shapiro \cite{Shapiro}. This problem is motivated by the problem of data acquisition in a submarine torpedo. The aim amounts to treating various environmental data coming from sensors located on the torpedo, that collect information which must be processed on a single processor. A single acquisition task can be described  as follows: a sensor of the torpedo emits a wave at a certain frequency (according to the data that must be  collected) which  propagates in the water and reflects back to the sensor. This acquisition task is divided into two sub-tasks:  
%BD : I've modified the following according to Reviewer 2 : the echo is the result of the reflexion of a pulse 
the first task consists in sending an ultrasound pulse while the second receives returning echo. Between them, there is an incompressible idle time which represents the spread of the echo under the water. Thus acquisition tasks may be assigned to coupled-tasks.

In order to use idle time, other sensors can send more echoes. However, the proximity of the waves causes disruptions and interferences. In order to handle information error-free, a compatibility graph between acquisition tasks is created. In this graph, which describes the set of tasks, we have an edge between two compatible tasks. A task is compatible with another if at least one of its sub-tasks can be executed during the idle time of another task. Given a set of coupled-tasks and such a compatibility graph, the aim is to schedule the coupled-tasks in order to minimize the time required for the completion of all the tasks.

\subsection{Notations}
First we present some common notations:
\begin{itemize}
\item Let $G$ be an undirected graph. We note $V(G)$ the set of its vertices and $E(G)$ the set of its edges;
\item  we note $n$ (resp. $m$) the  cardinality of set $V(G)$  (resp. $E(G)$); 
\item a $path$ is a non-empty graph $C$ with $V(C)=\{x_0,x_1,$ $\ldots,x_k \}$ and $E(C)\!=\!\{x_0x_1,\ldots,x_{k-1},x_k \}$, where all the $x_i$ are distinct;
\item the length of a path is the number of edges that the path uses.
\end{itemize}

Then, we introduce the notations relative to coupled-tasks we will use in the rest of the paper: we note $\mathcal{A}=\{A_1, A_2,\dots , A_n\}$ the set of $n$ coupled-tasks. Using the notation proposed by Shapiro \cite{Shapiro}, each task $A_i \in \mathcal{A}$ is composed of two sub-tasks $a_i$ and $b_i$. For clarity we use the same notations for the processing time of these tasks: $a_i$ and $b_i$ have processing time $a_i\in \mathbb{N}$ and $b_i\in \mathbb{N}$), and separated by a fixed idle time $L_i\in \mathbb{N}$ (see Figure \ref{single_coupled_task}). For each $i$ the second sub-task $b_i$ must start its execution exactly $L_i$ time units after the completion time of $a_i$. \\

According to the torpedo problem, a task may be started during the idle time of a running task if it uses another frequency, is not dependant on the execution of the running task (and reciprocally), or does not require to access the resources used by the running tasks. Formally, we say that two tasks $A_i$ and $A_j$ are \emph{compatible} if and only if we can execute at least a sub-task of $A_i$ during the idle time of $A_j$ (see Figure \ref{compatible_coupled_task}). On the other side, some tasks cannot be compatible due to previously cited reasons.

\begin{figure}[!hptb]
\begin{center}
\psfrag{a_i}{\small{$a_i$}}
\psfrag{b_i}{\small{$b_i$}}
\psfrag{a_j}{\small{$a_j$}}
\psfrag{b_j}{\small{$b_j$}}
\psfrag{L_i}{\small{$L_i$}}
\subfigure[]{\label{single_coupled_task}\includegraphics[height=.035\textwidth]{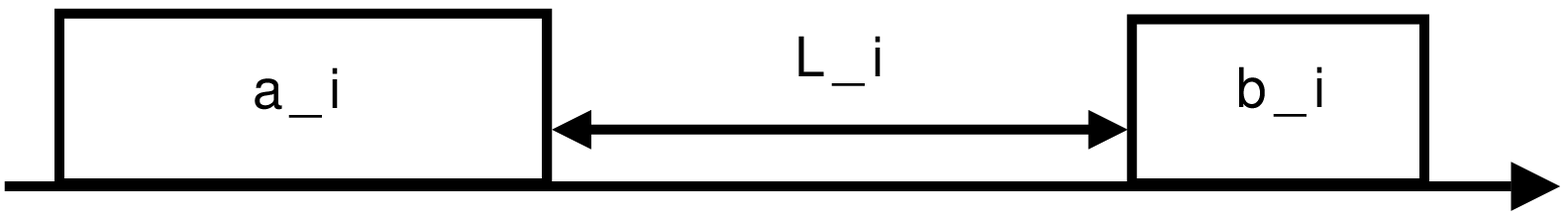}}
\subfigure[]{\label{compatible_coupled_task}\includegraphics[height=.035\textwidth]{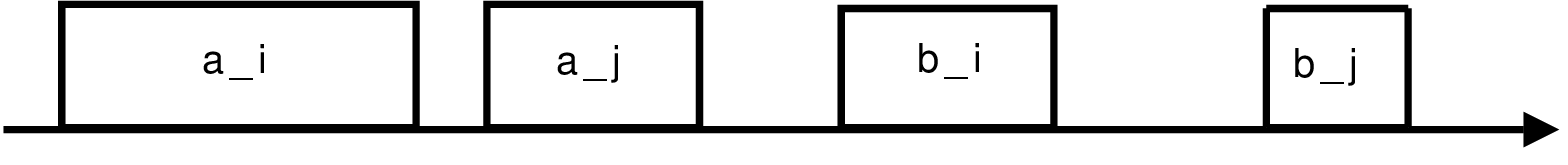}}
\caption{A single coupled-task and two compatible coupled-tasks.}
\end{center}
\end{figure}

\subsection{Main problem formulation}

We aim at scheduling a set  of coupled-tasks with compatibility constraints on a monoprocessor. 
%BD : I've modified " $G_c = (\mathcal{A},E(G_c))$, whose edges $E(G_c)$ " to avoid circular definition
The input of the general problem is described with the set $ \mathcal{A}=\{A_1, A_2,\dots , A_n\}$ of coupled-tasks and a compatibility graph $G_c$, with $V(G_c)=\mathcal{A}$ and $E(G_c)$ the edges which represent all pairs of compatible tasks, ie an edge exists between $A_i$ and $A_j$ if and if only $a_j$ can be scheduled between $a_i$ and $b_i$ (Fig. \ref{compatible_coupled_task}). Note that compatibility is symmetrical, thus  here $a_i$ could be scheduled between $a_j$ and $b_j$.

\begin{figure}[!hptb]
\begin{center}
\psfrag{A1}{\small{$A_1$}}
\psfrag{A2}{\small{$A_2$}}
\psfrag{A3}{\small{$A_3$}}
\psfrag{A21}{\small{$a_2$}}
\psfrag{A22}{\small{$b_2$}}
\psfrag{A11}{\small{$a_1$}}
\psfrag{A12}{\small{$b_1$}}
\psfrag{A31}{\small{$a_3$}}
\psfrag{A32}{\small{$b_3$}}
\psfrag{cg}{\small{Compatibility graph}}
\includegraphics[height=.07\textwidth]{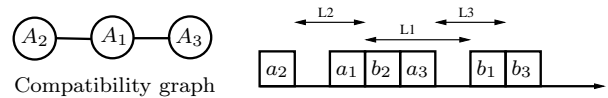}
\caption{Link between the compatibility graph and the scheduling}
\label{noncompatibles3}
\end{center}
\end{figure}

The solution of an instance consists in determining the starting time of each sub-task $a_i$ of each task $A_i\!\in\! \mathcal{A}$. The tasks have to be processed on a single processor while preserving the constraints given by the compatibility graph (see Figure \ref{noncompatibles3}). Formally,  we need to find a \emph{valid schedule} $\sigma : \mathcal{A} \rightarrow \mathbb{N}$  where the notation $\sigma(A_i)$ denotes the starting time of the task $A_i$. We use the following abuse of notation: $\sigma(a_i)=\sigma(A_i)$ (resp. $\sigma(b_i) =\sigma(A_i)+a_i + L_i$) denotes the starting time of the first sub-task $a_i$ (resp. the second sub-task $b_i$).\\

%BD : extraction of footnotes
Let $C_{max}=max_{A_i \in \mathcal{A}}(\sigma(A_i)+a_i+L_i + b_i)$ be the required time to complete all the tasks. Then the objective is to find a feasible schedule which minimizes $C_{max}$. We use the notation scheme $\alpha | \beta |\gamma$ proposed by Graham and al. \cite{GrahamLLR79}, where $\alpha$ denotes the environment processors, $\beta$ the characteristics of the jobs and $\gamma$ the criteria. The main problem denoted as $\Pi$ will be defined by:
\begin{equation*}
\Pi=1|coupled-task, (a_i,b_i,L_i), G_c |C_{max}
\end{equation*}

\subsection{Related work}

The problem of coupled-tasks has been studied in regard to different conditions on the values of $a_i$, $b_i$, $L_i$ for $1\leq i \leq n$, and precedence constraints \cite{BEKPTW09,Ahr,LBF09,op1997}. Note that, in the previous works, all tasks are compatible by considering a complete graph \cite{BEKPTW09,Ahr,LBF09,op1997}. Moreover, in presence of any compatibility graph, we find several complexity results \cite{GillesThese,troye,incom}, which are summarized in Table \ref{tabb1}. The notation $a_i=a$ implies that for all $1\leq i \leq n$, $a_i$ is equal to a constant $a \in \mathbb{N}$. This notation can be extended to $b_i$ and $L_i$ with the constants $b,L$ and $p \in \mathbb{N}$. 
 \begin{table}[!hptb]
  \begin{center}
  \ra{1.3} 
\begin{tabular}{@{}lll@{}}
\toprule
\textbf{Problem}  & \textbf{Complexity} & \textbf{ref} \\
  \hline
 $1\vert coupled\!-\!task, (a_i\!=\!b_i\!=\!L_i), G_c \vert C_{max}$  & $\mathcal{NP}$-complete&\cite{GillesThese}\\
 $1\vert coupled\!-\!task, (a_i\!=\!a,b_i\!=\!b,L_i\!=\!L), G_c \vert C_{max}$ & $\mathcal{NP}$-complete &\cite{GillesThese}\\
$ 1\vert coupled\!-\!task, (a_i\!=\!b_i\!=\!p,L_i\!=\!L), G_c \vert C_{max}$ &$ \mathcal{NP}$-complete &\cite{incom}\\ 
 $1\vert coupled\!-\!task, (a_i\!=\!L_i\!=\!p,b_i), G_c \vert C_{max}$ & $O(n^2m)$ &\cite{GillesThese}\\
$ 1\vert coupled\!-\!task, (a_i,b_i\!=\!L_i\!=\!p), G_c \vert C_{max}$ & $O(n^2m)$ &\cite{GillesThese}\\
\bottomrule 
\end{tabular}
\caption{Complexity for scheduling problems with coupled-tasks and compatibility constraints}
 \label{tabb1}
 \end{center}
\end{table}

\subsection{Contribution and organization of this paper}

Our work consists in measuring the impact of the compatibility graph on the complexity and approximation of scheduling problems with coupled-tasks on a monoprocessor. In this way, we focus our work on establishing the limits between polynomiality and $\mathcal{NP}$-completeness of these problems according to some parameters, when the compatibility constraints is introduced. In \cite{troye,incom}, we have studied the impact of the parameter $L$, and have shown that the problem $1\vert coupled-task, (a_i=b_i=p,L_i=L), G_c \vert C_{max}$ was $\mathcal{NP}$-complete as soon as $L\geq 2$, and polynomial otherwise. 
 
In this work, we complete complexity results with the study of other special cases according to the value of $a_i$ and $b_i$, and we propose several approximation algorithms for them. We restrict our study to a special case, by adding new hypotheses to the processing time and idle time of the tasks. For any task $A_i$, $i \in \{1,\ldots, n \}$, the processing time $a_i$ (resp. $b_i$) of sub-task $a_i$ (resp. $b_i$) is equal to a constant $a$ (resp. $b$), and the length of the idle time between $a_i$ and $b_i$ is $L$. Considering homogeneous tasks is a realistic hypothesis according to the tasks that the torpedo has to execute. Let $\mathbf{\Pi_1}$ be this new problem. Formally:
\begin{equation*}
\Pi_1 =1|coupled-task, (a_i=a,b_i=b,L_i=L), G_c |C_{max}
\end{equation*}

This paper is organized as follows: in section \ref{sect_cmplx}, we establish the complexity of $\Pi_1$ according to the values of $a$, $b$ and $L$, and we show that the problem is polynomial for any $L < a + b$; then we consider in the rest of the paper that $L = a + b$. In that case, the problem can be considered as a new graph problem we call \textsc{Minimum Schedule-Linked Disjoint-Path Cover} (\MinSLDPCF). We present the proof of $\mathcal{NP}$-completeness of \MinSLDPC and we conclude this section by the study of a specific sub-case with $a=b=L/2$. In Section \ref{sect-approx}, we show that \MinSLDPC is immediately $2$-approximated by a simple approach: we design a polynomial-time approximation algorithm with performance guarantee lower than $\frac{3}{2}$. In fact, we show that the approximation ratio obtained by this algorithm is between  $\frac{3}{2}$ and $\frac{5}{4}$, according to the values of $a$ and $b$.  The last section is devoted to the study of $\Pi_1$ for some particular topology of the graph $G_c$. First we present a well-known graph problem, \textsc{Minimum Disjoint-Path Cover}, (\MinDPCF). Then we show the relation between \MinDPC and \MinSLDPC and evaluate how results from the first one can be applied to solve the second problem on specific topologies. This implies  the reduction of the performance ratio we can obtain on some restricted instances from  $\frac{3}{2}$ to $\approx 1.37$.

\section{Computational complexity}
\label{sect_cmplx}

First, we prove  that $\Pi_1$ is polynomial when $L < a + b$: it is obvious that a maximum matching in the graph $G_c$ gives an optimal solution. Indeed, during the idle time $L$ of a coupled-task $A_i$, we can process at most one sub-task $a_j$ or $b_k$. Since the idle time $L$ is identical, so it is obvious that finding an optimal solution consists in computing a maximum matching.
Thus, the problem $1|coupled-task, (a_i\!=\!a,b_i\!=\!b,L_i\!=\!L<a+b), G_c |C_{max}$ admits a polynomial-time algorithm with complexity $O(m\sqrt(n))$ where $n$ is the number of tasks and $m$ the number of edges of $G_c$ (see \cite{schrijver04}).\\

The rest of the paper is devoted to the case $L = a + b$. Without loss of generality, we consider the case\footnote{The results we present here can be symmetrically extended to  instances with  $b > a$.} of $b < a$. The particular case $b = a$ will be discussed in subsection \ref{a=b}.

\subsection{From a scheduling problem to a graph problem}

Let us consider a valid schedule $\sigma$ of an instance $(\mathcal{A}, G_c)$ of $\Pi_1$ with $b < a$, composed of a set of coupled-tasks $\mathcal{A}$ and a compatibility graph $G_c$. For a given task $A_i$, at most two sub-tasks may be scheduled between the completion time of $a_i$ and the starting time of $b_i$, and in this case the only available schedule consists in executing a sub-task $b_j$ and a sub-task $a_k$ during the idle time $L_i$ with $i\neq j \neq k$ such that $\sigma(b_j)=\sigma(a_i)+a$ and $\sigma(a_k)=\sigma(a_i)+a+b$. Figure \ref{two_task_included_in_a_third} shows a such configuration. 

\begin{figure}[!hptb]
\begin{center}
\psfrag{a_i}{\small{$a_i$}}
\psfrag{b_i}{\small{$b_i$}}
\psfrag{a_j}{\small{$a_j$}}
\psfrag{b_j}{\small{$b_j$}}
\psfrag{a_k}{\small{$a_k$}}
\psfrag{a}{\small{$a$}}
\psfrag{b}{\small{$b$}}
\psfrag{b_k}{\small{$b_k$}}
\psfrag{L_i}{\small{$L_i$}}
\includegraphics[height=.075\textwidth]{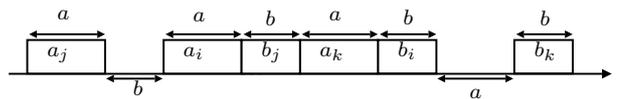}
\caption{At most $2$ sub-tasks may be scheduled between $a_i$ and $b_i$}
\label{two_task_included_in_a_third}
\end{center}
\end{figure}

We can conclude that any valid schedule $\sigma$ can be viewed as a partition $\{T_1, T_2, \dots, T_k\}$ of $\mathcal{A}$, such that for any $T_i$ the subgraph $P_i=G_c[T_i]$ of $G_c$ induced by vertices $T_i$ is a path (here, isolated vertices are considered as paths of length $0$). Clearly, $\{P_1, P_2\dots P_k\}$ is a partition of $G_c$ into vertex-disjoint paths. Figure \ref{the_main_example} shows an instance of $\Pi_1$ (Figure \ref{main_example_GC}), a valid schedule (Figure \ref{main_example}) - not necessarily an optimal one -, and the corresponding partition of $G_c$ into vertex-disjoint paths (Figure \ref{main_example_GC_partition}).

\begin{figure}[!hptb]
\begin{center}
\psfrag{a_1}{\scriptsize{$a_1$}}
\psfrag{b_1}{\scriptsize{$b_1$}}
\psfrag{a_2}{\scriptsize{$a_2$}}
\psfrag{b_2}{\scriptsize{$b_2$}}
\psfrag{a_3}{\scriptsize{$a_3$}}
\psfrag{b_3}{\scriptsize{$b_3$}}
\psfrag{a_4}{\scriptsize{$a_4$}}
\psfrag{b_4}{\scriptsize{$b_4$}}
\psfrag{a_5}{\scriptsize{$a_5$}}
\psfrag{b_5}{\scriptsize{$b_5$}}
\psfrag{a_6}{\scriptsize{$a_6$}}
\psfrag{b_6}{\scriptsize{$b_6$}}
\psfrag{a_7}{\scriptsize{$a_7$}}
\psfrag{b_7}{\scriptsize{$b_7$}}
\psfrag{A_1}{\scriptsize{$A_1$}}
\psfrag{A_2}{\scriptsize{$A_2$}}
\psfrag{A_3}{\scriptsize{$A_3$}}
\psfrag{A_4}{\scriptsize{$A_4$}}
\psfrag{A_5}{\scriptsize{$A_5$}}
\psfrag{A_6}{\scriptsize{$A_6$}}
\psfrag{A_7}{\scriptsize{$A_7$}}
\subfigure[An input graph $G_c$ with $7$ tasks]{\label{main_example_GC}\includegraphics[height=.14\textwidth]{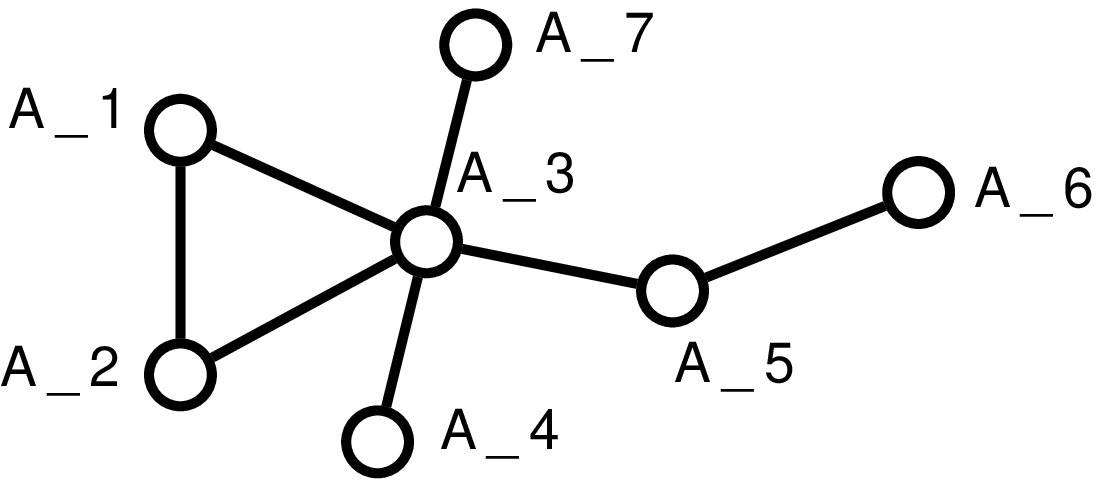}}
\hspace{1cm}
\subfigure[A vertex-disjoint partition]{\label{main_example_GC_partition}\includegraphics[height=.14\textwidth]{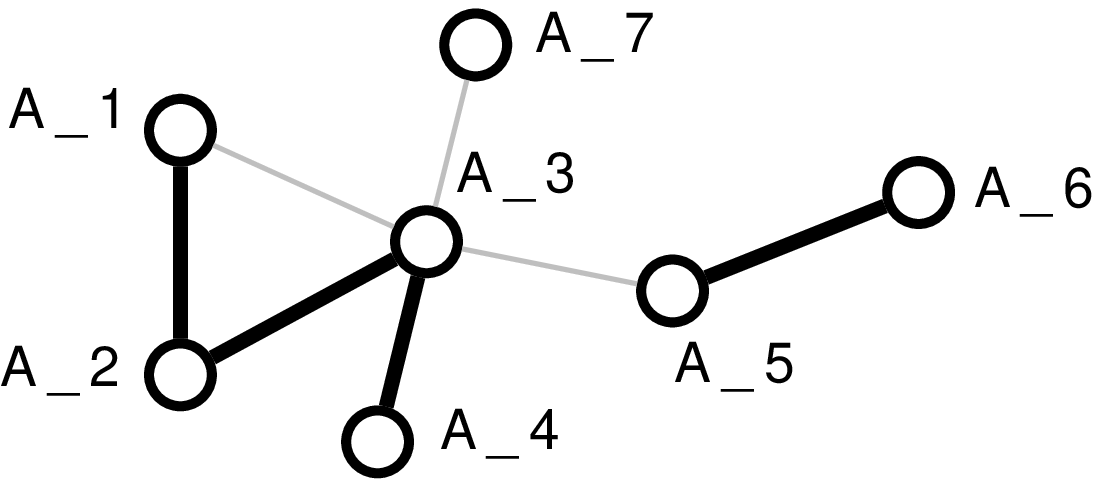}}
\subfigure[A schedule on a monoprocessor]{\label{main_example}\includegraphics[height=.052\textwidth]{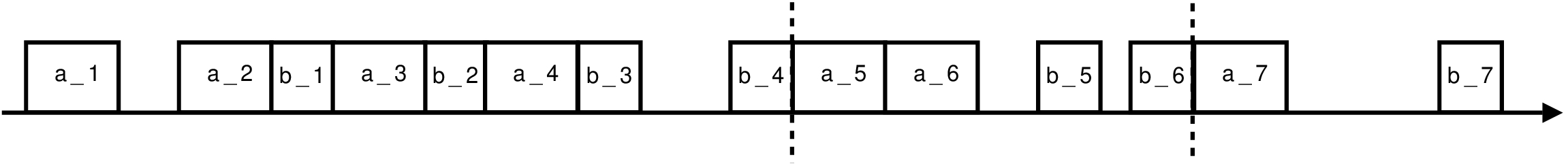}}
\caption{Relation between a schedule and a partition into vertex-disjoint paths}
\label{the_main_example}
\end{center}
\end{figure}

For a given feasible schedule $\sigma$, let us analyse the relation between the length of the schedule $C_{max}$ and the corresponding partition  $\{P_1, P_2, \dots P_k\}$ into vertex-disjoint paths. Clearly, we have $C_{max}=t_{seq}+t_{idle}$ where $t_{seq} = n(a+b)$ and $t_{idle}$ is the inactivity time of the processor. Since $t_{seq}$ is fixed for a given instance, $t_{idle}$ obviously depends on the partition. We propose the following lemma:

\begin{lemma}
Let $\{P_1, P_2, \dots P_k\}$ be the partition of vertex-disjoint paths corresponding to a schedule $\sigma$.
\begin{enumerate}
\item A path of length $0$ corresponds to a single task scheduled in $\sigma$, $t_{idle}$ is incremented by $L=a+b$;
\item for any path of length $1$,  $t_{idle}$ is increased by $a$;
\item for any path of length strictly greater than $1$, $t_{idle}$ is incremented by $(a+b)$.
\end{enumerate}
\end{lemma}

\begin{proof}
Point $1$ is obvious. Fig. \ref{cost_general} illustrates points $2$ and $3$: a path of length $1$ represents $2$ tasks that may be imbricated as on Figure  \ref{cost_1}. Paths of length strictly greater than $1$ represent more than two tasks. These tasks can be scheduled in order to get an idle time of length $b$ at the beginning of the schedule and one of length $a$ at the end of it (as on Figure \ref{cost_2}). The reader could check there is no other way to imbricate tasks in order to reduce the idle time for paths of any length. 

\begin{figure}[!hptb]
\begin{center}
\psfrag{a_i}{\small{$a_i$}}
\psfrag{b_i}{\small{$b_i$}}
\psfrag{a_j}{\small{$a_j$}}
\psfrag{b_j}{\small{$b_j$}}
\psfrag{a_k}{\small{$a_k$}}
\psfrag{b_k}{\small{$b_k$}}
\psfrag{idle time}{\scriptsize{$idle$ $time$}}
\psfrag{a}{\small{$a$}}
\psfrag{b}{\small{$b$}}
\psfrag{a-b}{\small{$a-b$}}
\psfrag{a+b}{\small{$a+b$}}
\psfrag{A_i}{\small{$A_i$}}
\psfrag{A_j}{\small{$A_j$}}
\psfrag{A_k}{\small{$A_k$}}
\psfrag{L_i}{\small{$L_i$}}
\subfigure[Chains of length $1$ increase $t_{idle}$ by $a$]{\label{cost_1}\includegraphics[height=.12\textwidth]{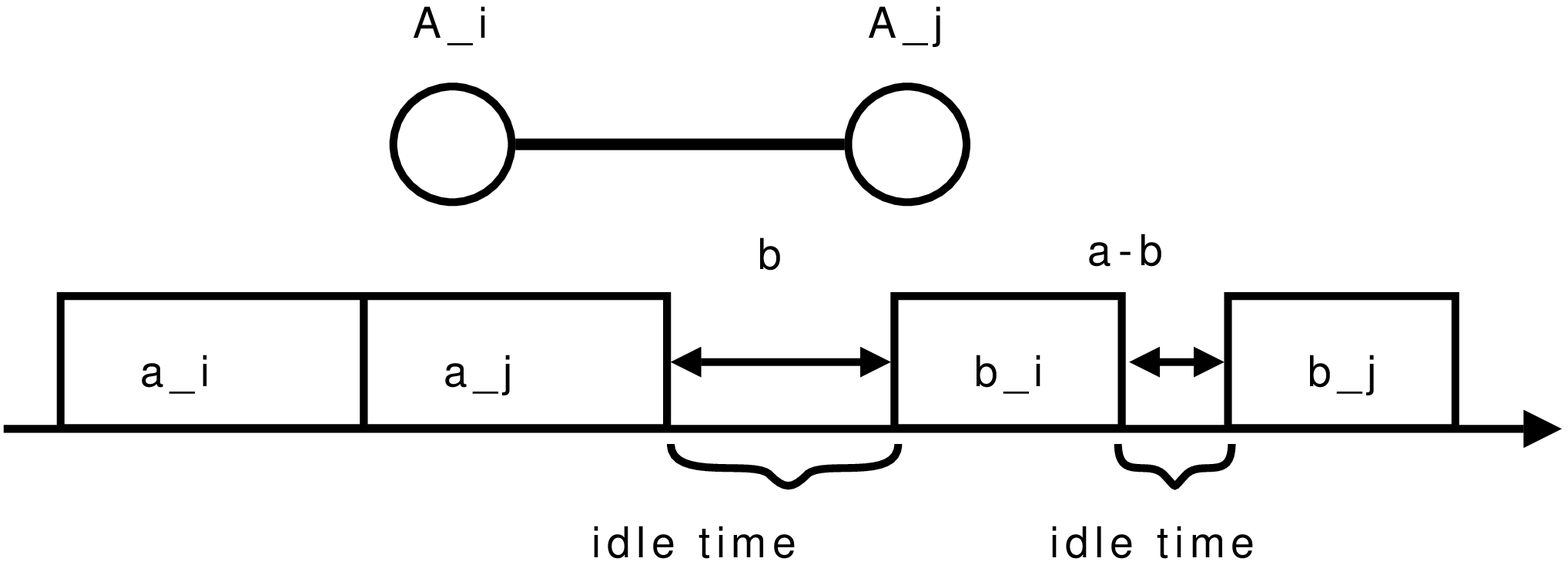}}
\subfigure[Chains of length $>2$ increase $t_{idle}$ by $a+b$]{\label{cost_2}\includegraphics[height=.11\textwidth]{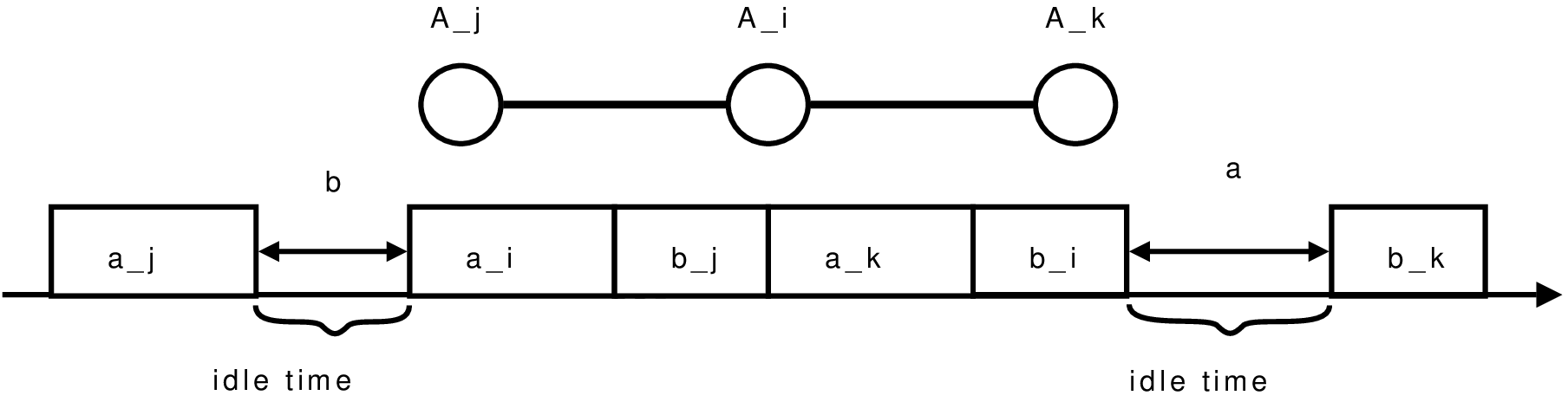}}

\caption{Impact of the length of the paths on the idle time}
\label{cost_general}
\end{center}
\end{figure}

\end{proof}

Thus, there exists a link between finding an optimal schedule and a graph problem which is called \textsc{Minimum Schedule-Linked Disjoint-Path Cover (\MinSLDPCF)} defined in Table \ref{RefMinSLDPC}:
\begin{table}[!hptb]
\begin{tabular}{|p{8cm}|}
\hline
\textbf{Instance:} a graph $G=(V, E)$ of order $n$, two natural integers $a$ and $b$,  $b < a$.\\
\textbf{Result:} a partition $\mathcal{P}$ of $G$ into vertex-disjoint paths (can be of length $0$) \\
\textbf{Objective:} Minimize $n(a+b) + \sum_{p\in \mathcal{P}}w(p)$ where $w : \mathcal{P} \rightarrow \mathbb{N}$ is a cost function with  $w(p) = a$ if and only if $|E(p)|=1$, and $w(p) = a+b$ otherwise.\\
\hline
\end{tabular}
\caption{ \textsc{Minimum Schedule-Linked Disjoint-Path Cover (\MinSLDPCF)}}
\label{RefMinSLDPC}
\end{table}
Clearly, \MinSLDPC is equivalent to $\Pi_1$ with $b < a$ and $L=a+b$, and can be viewed as the graph problem formulation of a scheduling problem. In any solution, each path increments the cost of idle time by at least $a$ (when the path has a length $1$), and at most $a+b < 2a$. So, we can deduce that an optimal solution to \MinSLDPC consists in finding a partition $\mathcal{P}$ with a particular cardinality $k^*$, and a maximal number of paths of length $1$ among all possible $k^*$-partitions. The following immediate theorem establishes the complexity of \MinSLDPCF:

\begin{theorem}
\MinSLDPC is an $\mathcal{NP}$-hard problem.
\end{theorem}
\begin{proof}
We consider the decision problem associated to \MinSLDPCF. We will prove that the 
the problem of deciding whether an instance of
\SLDPC has a schedule
of length at most $(n+1)(a+b)$ is ${\mathcal{NP}}$-complete.
Our proof is based on the polynomial-time transformation \textsc{Hamiltonian Path} $ \propto$ \SLDPCF. We keep the graph and the vertices is the task to schedule.
Let us consider a graph $G$.

 This transformation can be clearly computed in polynomial time.
\begin{bulletitemize}

\item Assume that the length of the optimal schedule is $C^{opt}_{max}=(n+1)(a+b)$. We will prove that the graph $G$ possess a Hamiltonian path i.e. $k=1$. Recall first that $k$ is the number of partition and that   $b <a$. We know, from the previous discussion,  that $t_{seq}=n(a+b)$ and that a chain of length one increase $t_{idle}$ by $a$ (see illustration given by Figures \ref{cost_2} and \ref{cost_1}).  It is clear that the graph $G$ must be covered by paths of different lengths.

Suppose that the graph $G$ is covered by $k$ paths with $k_1$ paths of length one (the set of these paths is denoted by $P_1$), and $k_2$ paths of lenght greater than one (resp. by $P_{\geq 2}$).

So the length of schedule given by this covering is:
\begin{eqnarray*}
C^{h}_{max}&=&\stackrel{\mbox{\scriptsize{processing times}}}{\overbrace{n(a+b)}}\!+\!\stackrel{\mbox{\scriptsize{idle time for }} P_{1} }{\overbrace{a \times k_1}}\!+\!\stackrel{\mbox{\scriptsize{idle time for }} P_{\geq 2}}{\overbrace{(a+b)k_2}}\\
&=& (n+k_2)(a+b)+k_1a >C^*_{max} \\
& & \mbox{if } (k_1 \!\neq\! 0 \mbox{ and } k_2 \!\geq\! 1) \ \ \mbox{or } (k_1 \!>\! 1 \mbox{ and } k_2 \!=\! 0)\\
\end{eqnarray*}

Thus, the only schedule requiring exactly $(n+1)(a+b)$ units of time implies that the graph possess a Hamiltonian path i.e. $k_1=0$ and $k_2=1$.
\item Reciprocally, we suppose that the graph $G$ possess a Hamiltonian path, we will prove the existence of a schedule of length $C_{max}=(n+1)(a+b)$.  

$G$ contains an Hamiltonian path, we can deduce a schedule with $C_{max}=(n+1)(a+b)$: as $t_{seq}=n(a+b)$, $t_{idle}$ must be equal to $(a+b)$, which is possible if and only if the schedule is represented with only one chain.
\end{bulletitemize}

\end{proof}

\subsection{A particular case $\Pi_2 \!:\! 1\vert coupled-task,(a_i\!=\!b_i\!=\!p,$ $L_i\!=\!L\!=\!2p), G_c \vert C_{max}$}
\label{a=b}

In this subsection only, we suppose that both sub-tasks are equal to a constant $p$ and that the inactivity time is equal to a constant $L=2p$.

The previous proof cannot be used for this case. Indeed the structure of these tasks allows to schedule three compatible tasks together without idle time (see Figure \ref{triangle}). Another solution consists in covering vertices of $G_c$ by triangles and paths (length $0$ allowed), where we minimize the number of paths and then maximize the number of path of length $1$.

This problem is a generalization of \textsc{Triangle Packing} \cite{GJ79} since an optimal solution without idle time consists in partitioning into triangles the vertices of $G_c$. This problem is well known to be $\mathcal{NP}$-complete, and leads to the $\mathcal{NP}$-completeness of problem $\Pi_2$.

\begin{figure}[!hptb]
	\begin{center}
\psfrag{a_1}{\small{$a_1$}}
\psfrag{a_2}{\small{$a_2$}}
\psfrag{a_3}{\small{$a_3$}}
\psfrag{b_1}{\small{$b_1$}}
\psfrag{b_2}{\small{$b_2$}}
\psfrag{b_3}{\small{$b_3$}}
\psfrag{A_1}{\small{$A_1$}}
\psfrag{A_2}{\small{$A_2$}}
\psfrag{A_3}{\small{$A_3$}}
\psfrag{2p}{\small{$2p$}}
\includegraphics[width=0.45\textwidth]{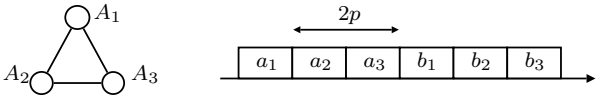}
\caption{Illustration of a schedule without idle time}
\label{triangle}
\end{center}
\end{figure}

A correct approximation algorithm for this problem is an algorithm close to the general case. Indeed, finding an optimal solution to this problem amounts to finding  a covering of the graph $G_c$ by triangles and paths, which minimize the idle time. In the following section, we will develop an efficient polynomial-time approximation algorithm for the general problem \MinSLDPCF.

\section{Approximation algorithm for \MinSLDPCF}
\label{sect-approx}

Notice that the following algorithm, executing sequentially the tasks, admits a ratio equal to two.

We also develop a polynomial-time $\frac{3}{2}$-approximation algorithm based on a maximum matching in the graph $G_c$. In fact, we show that this algorithm has an approximation ratio of at most $\frac{3a+2b}{2a+2b}$, which leads to a ratio between  $\frac{3}{2}$ and $\frac{5}{4}$ according to the values of $a$ and $b$ (with $b < a$). This result, which depends on the values $a$ and $b$, will be discussed in Section \ref{sect-approx-particular}, in order to propose a better ratio on some class of graphs.\\

For any instance of \MinSLDPCF, an optimal schedule has a length $C^{opt}_{max}\!=\!t_{seq}\!+\!t^{opt}_{idle}$ where $t_{seq}\!=\!n(a\!+\!b)$. 

\begin{remark}
For any solution of  length $C_{max}$, we necessarily have\footnote{The equality is obtained when the graph $G_c$ possesses an hamiltonian path, otherwise we need at least two paths to cover $G_c$ (where $G_c$ is not only an edge), which leads to increase $t_{idle}$ by at least $2a\geq a+b$ units of time.} $t_{idle}\geq (a+b)$ and also\footnote{The worst case consists in executing tasks sequentially without scheduling any sub-task $a_j$ or $b_j$ of task $A_j$ during the idle time of a task $A_i$.} $t_{idle}\leq n(a+b)$. Then, for any solution $h$ of \MinSLDPC we have a performance ratio $\rho(h)$ such that:
\begin{equation}
\rho(h) \leq \dfrac{C_{max}^h}{C_{max}^{opt}} \leq \dfrac{2n(a+b)}{(n+1)(a+b)} < 2.
\end{equation}

Indeed, we have $C_{max}^{opt} \geq T_{seq}=n(a+b)+(a+b)$
\end{remark}

In the following, we develop a polynomial-time approximation algorithm based on a maximum matching in the graph $G_c$, with performance guaranty in $[\frac{5}{4} , \frac{3}{2}]$ according to the values of $a$ and $b$.\\

Let $I$ be an instance of our problem. An optimal solution is a disjoint-paths cover. The $n$ vertices are partitioned in three disjoint sets: $n_1$ uncovered vertices, $n_2$ vertices covered by $\alpha_2 =\frac{n_2}{2}$  paths of length $1$, and $n_3$ vertices covered by exactly $\alpha_3$ paths of length strictly greater than $1$ (see illustration Figure \ref{optimal_cover}). The cost of an optimal solution is equal to the sum of sequential time  and idle time:
\begin{eqnarray*}
C_{max}^{opt}&=&\stackrel{\mbox{\scriptsize{processing times}}}{\overbrace{n(a+b)}}+\stackrel{\mbox{\scriptsize{idle time for matched vertices}} }{\overbrace{\frac{n_2}{2}a}}\\
&+&\hspace{-8mm}\stackrel{\mbox{\scriptsize{idle time for isolated vertices}}}{\overbrace{(a+b)n_1}}+\stackrel{\mbox{\scriptsize{idle time for path of length $>1$}}}{\overbrace{(a+b)\alpha_3}}
\end{eqnarray*}

\begin{figure}[!hptb]
\begin{center}
\psfrag{ldot}{\scriptsize{$\ldots$}}
\psfrag{n1}{\scriptsize{$n_1$ vertices}}
\psfrag{alpha3}{\scriptsize{$n_3$}}
\psfrag{vertices}{\scriptsize{vertices}}
\psfrag{n2}{\scriptsize{$n_2$ vertices}}
\psfrag{n3}{\scriptsize{$\alpha_3$ paths}}
\psfrag{ldots}{\scriptsize{$\ldots$}}
\includegraphics[height=.2\textwidth]{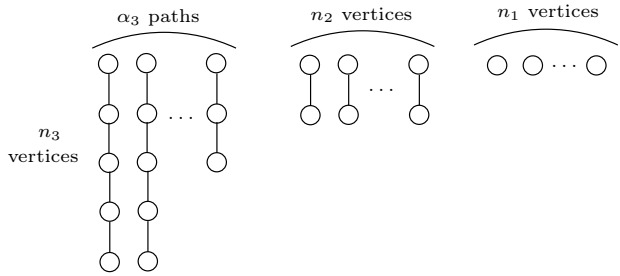}
\caption{Illustration of the optimal solution on an instance $I$}
\label{optimal_cover}
\end{center}
\end{figure}

Now, we propose a polynomial-time approximation algorithm with non trivial ratio on an instance $I$. This algorithm is based on a maximum matching in $G_c$ in order to process  two coupled-tasks at a time. For two coupled-tasks $A_i$ and $A_j$ connected by an edge of the matching, we obtain an idle time of length $a$ (see Figure \ref{cost_1}).

Let $M^*$ be the cardinality of a maximum matching. In the worst case, the $\alpha_3$ paths are all odd and a matching of the paths leaves $\alpha_3$ isolated vertices. So, we have by hypothesis:
 \begin{equation}
   M^* \geq \frac{n_2}{2}+(\frac{n_3}{2}-\alpha_3)=\Gamma\ \textnormal{(worst case)} % =m t'es sur ????
 \end{equation}

Indeed, based on the decomposition given by the Figure \ref{optimal_cover}, we may deduce a matching within this cardinality: the $\alpha_2$ paths of length one are contained in the matching;  for any path of length greater than one, we include odd edges to the matching. In the worst case, all $\alpha_3$ paths have odd length and the number of uncovered nodes is $\alpha_3$.
 If the tasks of the matching are processed first and the isolated vertices in second, the length of schedule is\footnote{Idle time $1$ (resp. $2$) represents the idle time for matched vertices (resp. isolated vertices).}:
\begin{equation*} % BD : j'ai remplace m par \frac{n_2}{2}
C_{max}^h  \leq \stackrel{\mbox{\scriptsize{processing times}}}{\overbrace{n(a+b)}}+\stackrel{\mbox{\scriptsize{idle time $1$}} }{\overbrace{a \times \Gamma}}+\stackrel{\mbox{\scriptsize{idle time $2$}}}{\overbrace{(a+b)(\alpha_3+n_1)}}
\end{equation*}

Since a optimal length is $C_{max}^{opt}\!=\! n(a\!+\!b)\!$ $+n_1(a\!+\!b)\!+\frac{n_2}{2}a\!+\!\alpha_3(a\!+\!b)$, we obtain, with the last equation involved $C_{max}^h $, the following ratio  of the polynomial-time approximation algorithm:

\begin{eqnarray*}
C_{max}^h  &\leq& C_{max}^{opt}+(\frac{n_3-\alpha_3}{2})a\\
\rho(h) &\leq&1+ \frac{(\frac{n_3-\alpha_3}{2})a}{n(a\!+\!b)\!+\!n_1(a\!+\!b)\!+\frac{n_2}{2}a+\!\alpha_3(a\!+\!b)}\\
 \rho(h) &\leq& 1+\frac{(\frac{n_3-\alpha_3}{2})a}{(n+n_1+\alpha_3)(a\!+\!b)\!+\frac{n_2}{2}a}\\
\rho(h) &\leq& 1+\frac{\frac{n_3}{2}a}{(n+n_1)(a\!+\!b)\!+\frac{n_2}{2}a} \mbox{\small{, max obtained for $\alpha_3=0$}}\\
\rho(h) & \leq & 1+\frac{\frac{n_3}{2}a}{n(a+b)} \leq 1+\frac{\frac{n}{2}a}{n(a+b)} \mbox{\small{, since $n_3 \leq n$}}\\
\rho(h) & \leq & 1+\frac{a}{2(a+b)}= \frac{3a+2b}{2a+2b}\\
\end{eqnarray*}

\section{Instances with particular topologies}
\label{sect-approx-particular}

We conclude this work by a study of \MinSLDPC when $G_c$ admits a particular topology. First we present a related problem: 
 \textsc{Minimum Disjoint Path Cover problem} (\MinDPCF). This problem has some interesting results on restricted topologies. We establish a link between \MinDPC and \MinSLDPC and we show that finding a $\rho_{dpc}$-approximation for \MinDPC on $G_c$ allows to find a strategy with performance ratio $\rho_{sldpc} \leq \min \{ \rho_{dpc}\!\times\!(\frac{a+b}{a}), \frac{3a+2b}{2a+2b} \}$. This leads to propose, independently from the values $a$ and $b$, a $\frac{1+\sqrt{3}}{2}$-approximation for \MinSLDPC when \MinDPC can be polynomially solved on $G_c$. 

\subsection{A related problem:  \MinDPCF}

The graph problem  \MinSLDPC is very close to the well-known problem  \textsc{Minimum Disjoint Path Cover (\MinDPCF)} which consists in covering the vertices of a graph with a minimum number of vertex-disjoint paths\footnote{Sometimes referenced as the \textsc{Path-Partition} problem (PP).}. This problem has been studied in depth in several graph classes: it is known that this problem is polynomial on cographs \cite{NOZ03}, blocks graphs and bipartite permutation graphs \cite{SSSSPR93}, distance-hereditary graph \cite{HC07}, and on interval graphs \cite{RAPR90}. In \cite{BCM77} and \cite{GHS75}, the authors have proposed (independently) a polynomial-time algorithm in the case where the graph is a tree. Few years later, in \cite{K76} the authors showed that this algorithm can be implemented in linear time. Among the other results, there is a polynomial-time algorithm for the cacti \cite{MW88}, and another for the line graphs of a cactus \cite{DM04}. In circular-arc graphs the authors \cite{HC06} have proposed an approximation algorithm of complexity $O(n)$, which returns an optimal number of paths to a nearly constant additive equal to $1$.

The problem \MinDPC is directly linked to \textsc{Hamiltonian Completion} \cite{GJ79}, which consists in finding the minimum number of edges, noted $HC(G)$, that must be added to a given graph $G$, in order to make it hamiltonian (to guarantee the existence of a Hamiltonian cycle). It is known that if $G$ is not hamiltonian, then the cardinality of a minimum disjoint path cover is clearly equal to $HC(G)$.

The dual of \MinDPC is  \textsc{Maximum Disjoint-Path Cover} \cite{GJ79}. It consists in finding in $G$ a collection of vertex-disjoint paths of length at least $1$, which maximises the edges covered in $G$. This problem is known to be $\frac{7}{6}$-approximable \cite{BK06}.

\subsection{Relation between \MinDPC and \MinSLDPCF}

From the literature, we know that \MinDPC is polynomial on trees \cite{GHS75,K76,BCM77}, distance-hereditary graphs \cite{HC07}, bipartite permutation graphs \cite{SSSSPR93}, cactis \cite{MW88} and many others classes. There are currently no result about the complexity of \MinSLDPC on such graphs: since the values of of $a$ and $b$ have a high impact, techniques used to prove the polynomiality of \MinDPC cannot be adapted to prove the polynomiality of \MinSLDPCF. Despite all our effort, the complexity of \MinSLDPC remains an open problem. However using known results on \MinDPCF, we show how the approximation ratio can be decreased for these class of graphs. We propose the following lemma:
\begin{lemma}
If \MinDPC can be solved polynomial computation time, then there exists a polynomial-time $\frac{(a+b)}{a}$-approximation for \MinSLDPCF.
\end{lemma}

\begin{proof}
Let $I_1=(G)$ be an instance of  \MinDPCF, and $I_2=(G,a,b)$ an instance of \MinSLDPCF.
Let $\mathcal{P}^*_1$ be an optimal solution of \MinDPC of cost $|\mathcal{P}^*_1|$ , and  $\mathcal{P}^*_2$ an optimal solution of \MinSLDPC of cost $OPT_{sldpc}$. According to the definition of  \MinSLDPCF,  we have $|\mathcal{P}^*_2| \geq |\mathcal{P}^*_1|$\footnote{The best solution for \MinSLDPC is not necessarily a solution with a minimum cardinality of $k$.}. Since each path of a \MinSLDPC solution increments the cost of the solution by at least $a$, then we have:
\begin{eqnarray}
OPT_{sldpc} &\!=\!& \sum_{p\in \mathcal{P}^*_2} w(p) \!+\!n(a\!+\!b)\nonumber\\
& \!\geq\!&  a|\mathcal{P}^*_2|\!+\!n(a\!+\!b) \!\geq\!  a|\mathcal{P}^*_2| \label{eq_1}\\
& \Rightarrow &\frac{OPT_{sldpc}}{a} \geq |\mathcal{P}^*_2|
\label{eq_2}
\end{eqnarray}

Let us consider the partition $\mathcal{P}^*_1$  as a solution (not necessarily optimal) to the instance $I_2$ of \MinSLDPCF, and let us evaluate it cost. Since each path of a \MinSLDPC solution increments the cost of the solution by at most $a+b$, then we have:
\begin{eqnarray}
&\sum_{p\in \mathcal{P}^*_1}&w(p) +n(a+b)  \leq  (a+b)|\mathcal{P}^*_1|+n(a+b)\nonumber\\
& \leq& (a+b)|\mathcal{P}^*_2| +n(a+b) \nonumber\\
& \leq &  b|\mathcal{P}^*_2| + OPT_{sldpc}\qquad \ \ \  \textnormal{ according to (\ref{eq_1})}\nonumber\\
& \leq & \frac{b}{a}OPT_{sldpc}+OPT_{sldpc} \ \ \textnormal{according to (\ref{eq_2})}\nonumber\\
& \leq & \frac{a+b}{a}OPT_{sldpc}
\end{eqnarray}

\end{proof}

The same proof may be applied if there exists a $\rho_{dpc}$-approximation for \MinDPCF, and then there exists a $\rho_{dpc}\!\times\!(\frac{a+b}{a})$-approximation for \MinSLDPCF. Let us suppose that we know a constant $\rho_{dpc}$ such that there exists a $\rho_{dpc}$-approximation for \MinDPC on $G_c$. Let $S_1$ be the strategy, which consists in determining a $\rho_{dpc}\!\times\!(\frac{a+b}{a})$-approximation for \MinSLDPC from $\rho_{dpc}$, and $S_2$ the strategy, which consists in using the algorithm introduced in section \ref{sect-approx}. Clearly, $S_1$ is particularly relevant when $b$ is very small in comparison with $a$.  Whereas $S_2$ gives better ratio when $b$ is close to $a$. Both strategies are complementary along the value of $b$, which varies from $0$ to $a$. Choosing the best result between the execution of $S_1$ and $S_2$, gives a performance ratio $\rho_{sldpc}$ such that: 
\begin{equation}
\rho_{sldpc} \leq \min\big\{ \rho_{dpc}\times \left(\frac{a+b}{a}\right), \frac{3a+2b}{2a+2b}\big\}.
\end{equation}

Compared to executing $S_1$ only, this new strategy increases the obtained results if and only if $\rho_{dpc}$ is lower than $\frac{3}{2}$ (see Figure \ref{graphique_rho_approx}). \\

\begin{figure}[!hptb]
\begin{center}
\psfrag{a}{$a$}
\psfrag{a/2}{$\frac{a}{2}$}
\psfrag{a/4}{$\frac{a}{4}$}
\psfrag{3a/4}{$\frac{3a}{4}$}
\psfrag{0}{$0$}
\psfrag{1}{$1$}
\psfrag{rho}{$\rho_{dpc}$}
\psfrag{5/4}{$\frac{5}{4}$}
\psfrag{3/2}{$\frac{3}{2}$}
\psfrag{S1}{$S_1$}
\psfrag{S2}{$S_2$}
\psfrag{bound}{\small{bound}}
\psfrag{value of b}{\small{value of b}}
\psfrag{approximation ratio}{\small{approximation ratio}}
\includegraphics[height=.22\textwidth]{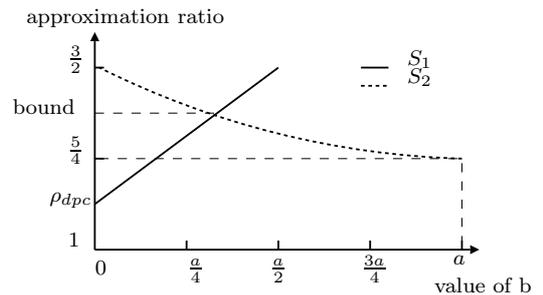}
\caption{Finding a good $\rho_{dcp}$-approximation helps to increase the results of Section \ref{sect-approx}}
\label{graphique_rho_approx}
\end{center}
\end{figure}

We propose the following remark, which is not good news: 
\begin{remark}
There is no $\rho_{dpc}$-approximation for \MinDPC in general graphs for some $\rho_{dpc} < 2$
\end{remark}

This result is a consequence of the Impossibility Theorem \cite{CP1995}. It can be checked by considering an ins\-tance of \MinDPC which has an hamiltonian path: the optimal solution of \MinDPC has cost $1$, thus any polynomial-time $\rho_{dpc}$ - approximation algorithm with $ \rho_{dpc}<2$ will return a solution of cost $1$, which is not allowed under the assumption that $\mathcal{P} \neq \mathcal{NP}$. This result also implies that constant-factor approximation algorithms for max-DPC do not necessarily give the same performance guarantees on min-DCP, since the best approximation ratio for max-DCP is $\frac{7}{6}$, which is lower than the inapproximability bound for min-DCP.\\

It is good news that \MinDPC is polynomial for some compatibility graphs such as trees \cite{GHS75,K76,BCM77}, distance-hereditary graphs \cite{HC07}, bipartite permutation graphs \cite{SSSSPR93}, cactis \cite{MW88} and many others classes; thus $\rho_{dpc}=1$. For all these graphs we obtain an approximation ratio of $\min \{ \frac{(a+b)}{a}, \frac{3a+2b}{2a+2b} \}$ which is maximal when $\frac{(a+b)}{a} = \frac{3a+2b}{2a+2b}$, i.e.:
\begin{eqnarray}
\frac{(a+b)}{a} = \frac{3a+2b}{2a+2b} \Leftrightarrow
- a^2+2ab+2b^2 = 0.
\end{eqnarray}

The only solution of this equation with $a$ and $b\geq 0$ is $a=b (1 + \sqrt{3})$. By replacing $a$ by this new value on $\frac{(a+b)}{a}$ or on $\frac{3a+2b}{2a+2b}$, we show that in the worst case the approximation ratio is reduced from $\frac{3}{2}$ down to $\frac{1+\sqrt{3}}{2}\approx 1.37$.

\section{Conclusion}

We investigate a particular coupled-tasks scheduling problem $\Pi_1$ in presence of a compatibility graph. We have shown how this scheduling problem can be reduced to a graph problem. We have proved that adding the compatibility graph leads to the $\mathcal{NP}$-completeness of $\Pi_1$, whereas the problem is obviously polynomial when there is a complete compatibility graph (each task is compatible with each other).  We have proposed a $\rho$-approximation of $\Pi_1$ where $\rho$ is between $\frac{3}{2}$ and  $\frac{5}{4}$ according to value of $a$ and $b$. We have also decreased the upper bound of $\frac{3}{2}$ down to $\approx 1.37$ on instances where \textsc{Minimum Disjoint Path Cover} can be polynomially  solved on the compatibility graph. 

As perspectives of this work, we plan to test the pertinence of our poly\-nomial-time approximation algorithm  through simulations in order to determine the average gap between the optimal solution and the results obtained with our strategy on significant instances. We also aim to classify  the complexity of other configurations,  especially $\Pi_1: 1 | coupled-task, (a_i=a, b_i=b, L_i=L) | C_{max}$ with a complete compatibility graph.

\section{Acknowledgements}

Thanks to our reviewers for the thorough review of this paper and many helpful comments and suggestions.

\nocite{}
\bibliographystyle{plain}
\bibliography{./newBib}

\end{document}